\documentclass[journal]{IEEEtran}

\usepackage{graphicx}
\usepackage{amsmath,amsfonts,amssymb,amsthm}
\usepackage{balance}
\usepackage{cite}
\usepackage{fixmath}
\usepackage[center]{caption}
\usepackage{subcaption}
\usepackage{multirow}
\newtheorem{prop}{Proposition}
\usepackage[english]{babel}
\usepackage{cite}
\usepackage{epstopdf}
\usepackage{amsmath,amssymb,amsthm,mathrsfs,amsfonts,dsfont}
\usepackage{algorithmic}
\usepackage{graphicx}
\usepackage{graphics}
\usepackage{color}
\usepackage{textcomp}
\usepackage{balance}
\usepackage[top=0.75in,bottom=1in,left=0.675in,right=0.675in]{geometry}

\newcommand{\site}{\boldsymbol{s}}

\begin{document}

\title{Radio Resource Dimensioning with Cox Process Based User Location Distribution}
\author{Ridha Nasri, 
	Jalal Rachad and
	Laurent Decreusefond
	
	\thanks{Ridha Nasri and Jalal Rachad are with Orange Labs, 40-48 Avenue de la R\'epublique, 92320 Ch\^atillon, France. Emails: \{ridha.nasri,jalal.rachad\}@orange.com}
	\thanks{Jalal Rachad and Laurent Decreusefond are with LTCI Telecom Paris-Tech, universit\'e Paris Saclay, 23 Avenue d'Italie,75013, Paris, France. Emails: jalal.rachad@telecom-paristech.fr, laurent.decreusefond@mines-telecom.fr}
}


\maketitle

\begin{abstract}
The upcoming fifth generation (5G) New Radio (NR) interface inherits many concepts and techniques from 4G systems such as the Orthogonal Frequency Division Multiplex (OFDM) based waveform and multiple access. Dimensioning 5G NR interface will likely follow the same principles as in 4G networks. It aims at finding the number of radio resources required to carry a forecast data traffic at a target users Quality of Services (QoS). The present paper attempts to provide a new approach of radio resources dimensioning considering the congestion probability, qualified as a relevant metric for QoS evaluation. We distinguish between the spatial random distribution of indoor users, modeled by a spatial Poisson Point Process (spatial PPP) in a typical area covered by a 5G cell, and the distribution of outdoor users modeled by a linear PPP generated in a random system of roads modeled according to a Poisson Line Process (PLP). Moreover, we show that the total requested Physical Resource Blocks (PRBs) follows a compound Poisson distribution and we attempt to derive the explicit expression of the congestion probability by introducing a mathematical tool from combinatorial analysis called the exponential Bell polynomials. Finally we show how to dimension radio resources, for a given target congestion probability, by solving an implicit relation between the necessary resources and the forecast data traffic expressed in terms of cell throughput. Different numerical results are presented to justify this dimensioning approach.   
   
\end{abstract}


\IEEEpeerreviewmaketitle

\section{Introduction}

\IEEEPARstart{R}{adio} dimensioning consists in assessing the network resources required to carry a predicted data traffic with a satisfactory QoS. This later is often summarized in some metrics such as the average user throughput or the target congestion probability. In contrast with some recent works, where the dimensioning exercise is performed to satisfy a minimum user throughput in the cell, we use the cell congestion as the target QoS, instead. Besides, dimensioning is performed assuming mobile users distributed in roads or located in buildings. The first kind of users are modeled by Cox point process driven by PLP whereas the second kind is described by a spatial PPP. Such mobile users are granted some radio resources, called Physical Resource Blocks (PRB), at each Time Transmit Interval (TTI) and according to a predefined scheduling algorithm. The choice of the scheduling algorithm is mainly related to the fairness level made between users, i.e., the way that resources are allocated to users according to their channel qualities and their priorities, defined by the operator \cite{trivedi2014comparison,sadr2009radio,yaacoub2012survey}.

\subsection{Related works}

Dimensioning approaches, resource allocation and scheduling algorithms have been widely addressed in literature for OFDMA access technology; see for instance \cite{shen2005adaptive,sadr2009radio,yaacoub2012survey,el2011dynamic,el2012stable,agarwal2008low,khattab2006opportunistic,decreusefond2012robust,blaszczyszyn2009dimensioning}. In \cite{shen2005adaptive}, an adaptive resource allocation for multiuser OFDM system, with a set of proportional fairness constraints guaranteeing the required data rate, has been discussed. Similarly, authors in \cite{sadr2009radio} surveyed different adaptive resource allocation algorithms and provided a comparison between them in terms of performance and complexity. Furthermore, OFDMA dimensioning has been always considered as a hard task because of the presence of elastic data services. It was provided in \cite{decreusefond2012robust} an analytical model for dimensioning OFDMA based networks with proportional fairness in resource allocation between users requiring different transmission rates. For a Poisson distribution of mobile users, authors in \cite{decreusefond2012robust} showed that the required number of resources in a typical cell follows a compound Poisson distribution. In addition, an upper bound of the blocking probability was given. Likewise in \cite{blaszczyszyn2009dimensioning}, authors have proposed a Downlink OFDMA dimensioning approach considering an Erlang's loss model and Kaufman-Roberts algorithm to evaluate the blocking probability. Also in \cite{karray2010analytical}, it has been proposed an analytical method to evaluate the QoS for Downlink OFDMA system considering real-time and elastic traffic with a dimensioning approach illustration.\\

Additionally, Different models for network geometry and user distributions can be found in \cite{nasri2016analytical, andrews2011tractable, chetlur2018coverage, choi2018analytical}. Stochastic geometry is a strong mathematical tool to model the spatial randomness of wireless communication and also the random tessellations of roads. In particular, authors in \cite{chetlur2018coverage} and \cite{choi2018analytical} considered vehicular-type communication systems where the transmitting and receiving nodes are distributed along roads and modeled by a linear PPP, while roads random tessellations are modeled by a PLP, i.e., the process of nodes is doubly stochastic. Such a model is known as Cox point process driven by PLP. Others models have been proposed in literature such as Manhatan model that uses a grid of horizontal and vertical streets, Poisson Voronoi Tessellations (PVT) and Poisson Delaunay Tessellations (PDT) \cite{wang2018mmwave,gloaguen2010analysis,voss2009distributional}. Manhattan model does not fit the irregularity of roads in urban and dense urban environment, while PVT and PDT could not lead to explicit analytical result. It seems that Cox point process driven by PLP is a relevant model for roads in urban environment that is gaining popularity recently and merits investigations when looking for performance analysis and dimensioning problems of wireless cellular communications.
 

\subsection{Contribution}

Compared to the existing works, the main contributions of this paper are:

\begin{description}
	
	\item[$\bullet$] We provide an analytical model to dimension OFDM based systems with a proportional fair resources' allocation policy. This dimensioning model is very useful for operators because it gives a vision on how they should manage the available spectrum. If the dimensioned number of resources exceeds the available one, the operator can, for instance,  aggregate fragmented spectrum resources into a single wider band in order to increase the available PRBs, or activate capacity improvement features like dual connectivity between 5G and legacy 4G networks, in order to delay investment on the acquisition of new spectrum bands. Moreover, the proposed model can be applied to the scalable OFDM based 5G NR with different subcarriers' spacing in order to enable different types of deployments and network topologies and support different use cases.\\
	\item[$\bullet$] Instead of considering only the random distribution of users in the cell often modeled by a spatial PPP, we consider two types of users: \textit{i)} indoor users distributed in buildings and modeled by a spatial PPP. \textit{ii)} for outdoor users (e.g., pedestrians or vehicular), we characterize at first the random distribution of roads in a typical cell coverage area by a PLP and then we consider the random distribution of users in this system of roads according to a linear PPP. This model allows the operator to evaluate and compare performances between outdoor and indoor environments in terms of required radio resources.\\ 
	\item[$\bullet$] We show that the total number of the requested PRBs follows a compound Poisson distribution and we derive the explicit formula of the congestion probability as a function of different system parameters by using a mathematical tool from combinatorial analysis called the exponential Bell polynomials. This metric is defined as the risk that the requested resources exceed the available ones. It is often considered primordial for operators when it comes to resources dimensioning since it is related to the guaranteed quality of service. Then by setting a target congestion probability, we show how to dimension the number of PRBs given a forecast cell throughput. To the best of our knowledge, the explicit formula of the congestion probability has never been derived in similar studies.

 \end{description}


\subsection{Paper organization}

The rest of this paper is organized as follows: In Section II, system models, including a short description of Poisson Line Process, are provided. Section III characterizes the proposed dimensioning model and provides an explicit expression of the congestion probability and an implicit relation between the number of required resources and the cell throughput. Numerical results are provided in Section IV. Section V concludes the paper.

\section{System model and notations}

Cellular networks modeling is often related to the network geometry, the shape of the cell, the association between cells and users and of course their spatial distributions. This latter is related to the geometry of the city where the studied cell area exists. The Geometry of the city, in turn, is linked to the spatial distribution of roads and buildings. Indoor users, which are distributed in buildings, are often modeled by a spatial PPP in $\mathbb{R}^2$. However, outdoor users (e.g., pedestrians or vehicular) are always distributed along roads. As we mentioned in the introduction, many models have been proposed in literature to model the spatial distribution of roads, such as Manhattan model, PVT, PDT and Poisson Line Process. In this work, we consider a combination of indoor and outdoor users in the studied cell area. Indoor users are distributed according to  spatial PPP and outdoor users are distributed along a random system of roads according to a Cox Point Process driven by PLP.

\subsection{Indoor users model}

 A PPP in $\mathbb{R}^2$ with intensity $\zeta$ is a point process that satisfies: \textit{i}) the number of points inside every bounded closed set $B \in \mathbb{R}^2$ follows a Poisson distribution with mean $\zeta|B|$, where $|B|$ is the Lebesgue measure on $\mathbb{R}^2$; \textit{ii}) the number of points inside any disjoint sets of $\mathbb{R}^2$ are independent \cite{nasri2015tractable}. Actually, spatial PPP has been widely used to model BSs and users locations in cellular network. In this work, indoor users are considered to be distributed in buildings according to a spatial PPP $\varphi$ of intensity $\kappa$, which means that their locations are uniformly distributed in the studied cell coverage area and their number follows a Poisson distribution. 

\subsection{Outdoor users model} 

As we mentioned previously, outdoor users are considered to be distributed along a random system of roads. To model the random tessellation of roads, we consider the so-called PLP which is mathematically derived from the spatial PPP. Instead of points, the PLP is a random process of lines distributed in the plane $\mathbb{R}^2$. Each line in $\mathbb{R}^2$ is parametrized in terms of polar coordinates ($r$,$\theta$) obtained from the orthogonal projection of the origin on that line, with $r$ $\in$ $\mathbb{R^+}$ and $\theta \in (-\pi, \pi]$. Now we can consider an application $T$ that maps each line to a unique couple ($r$,$\theta$), generated by a PPP in the half-cylinder $\mathbb{R^+} \times (-\pi,\pi]$; Fig. \ref{plp1}. The distribution of lines in $\mathbb{R}^2$ is the same as points' distribution in this half-cylinder; see \cite{chetlur2018coverage} and \cite{choi2017analytical} for more details.\\

In the sequel, we assume that roads are modeled by a PLP $\phi$ with roads' intensity denoted by $\lambda$. The number of roads that lie inside a disk $\site$ of radius $R$ is a Poisson random variable, denoted by $Y$. It corresponds to the number of points of the equivalent spatial PPP in the half-cylinder $[0,R] \times (-\pi,\pi]$ having an area of $2\pi R$. Hence, the expected number of roads that lie inside $\site$ is $\mathbb{E}(Y)=2 \pi \lambda R$. Then, conditionally on $\phi$ (i.e., conditionally on roads), outdoor users are assumed to be distributed on each road according to independent linear PPPs having the same intensity $\delta$. This model is known as Cox point process. The mean number of users on a given road $j$ is $\delta L_j$, with $L_j$ is the length of road $j$. Besides, the number of roads that lie between two disks of radius $R_1$ and $R_2$ respectively, with $R_1 \leqslant R_2$, is $2\pi \lambda (R_2-R1)$. Also, the number of distributed users in a road, parametrized by ($r$,$\theta$) and delimited by the two disks, is $2\delta (\sqrt{R_2^2 - r^2}-\sqrt{R_1^2 - r^2})$. Additionally, the average number of outdoor users in the disk of radius $R$ can be calculated using the equivalent homogeneous spatial PPP with intensity $\lambda \delta$ in the disk area.
For illustration, Fig. \ref{plp1} presents the line parametrization described above and Fig. \ref{plp2} shows a realization of a Cox Point Process driven by PLP.\\  

\begin{figure}[tb]
	\centering
	\includegraphics[width=0.45\textwidth]{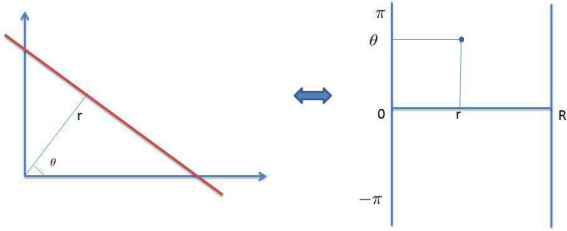}
	\caption{Line parametrization.}
	\label{plp1}
\end{figure}

\begin{figure}[tb]
	\centering
	\includegraphics[width=0.47\textwidth]{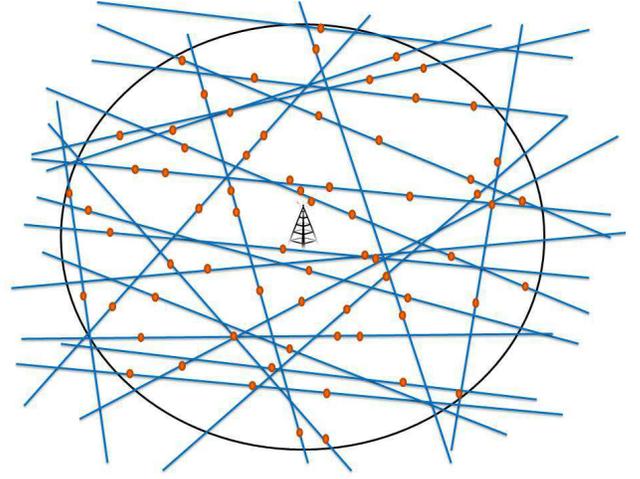}
	\caption{A realization of Cox Point Process driven by PLP.}
	\label{plp2}
\end{figure} 

 Additionally, we assume that outdoor and indoor users processes are independent and they form respectively two processes with intensities $\lambda \delta$ and $\kappa$. Therefore, the average number of users (outdoor and indoor), denoted by $u$, inside the cell coverage area can be calculated by 
 
\begin{equation}
u= ( \lambda \delta + \kappa) \pi R^2.
\label{nbrusers}
\end{equation}

Table \ref{tab:table1} summarizes the basic notations used in the article.

\begin{table}[h!]
	\begin{center}

		\begin{tabular}{l|c|r} 
			
			\textbf{Symbols} & \textbf{Definition} \\

			\hline
			$\varphi$ & spatial PPP of indoor users with intensity $\kappa$\\
			
			\hline
			$\phi$ & PLP of roads with intensity $\lambda$ \\
			
			\hline
			$\delta$ & The linear PPP intensity on each road (outdoor) \\
			
			\hline
			$r_j$ & The short distance between a road j and the origin \\
			
			\hline
			$Y$ &  number of roads that lie inside $\site$ \\
			
			\hline
			$\delta \lambda$ & Spatial PPP intensity on half-cylinder $[0,R] \times (-\pi,\pi]$  \\
			\hline
			$u$ & Average number of users in $\site$  \\
			\hline
		
		\end{tabular}
	\caption{Notations.}
	\label{tab:table1}
	\end{center}
\end{table}

\subsection{Network model}

We consider a circular cell $\site$ of radius $R$ with a base station (BS), denoted also $\site$ and positioned at its center, transmitting with a power level $P$. The received power by a user located at distance $x$ from $\site$ is $Px^{-2b}/a$, where $2b$ is the path loss exponent and $a$ is a propagation parameter that depends on the type of the environment (outdoor, indoor). We assume that BS $\site$ allocates PRBs to its users at every TTI (e.g., 1 ms). Each PRB has a bandwidth denoted by $W$ (e.g., $W=$180kHz for scalable OFDM with subcarriers spacing of 15kHz).\\

Active users in the cell compete to have access to the available dimensioned PRBs. Their number is denoted by $M$. The BS allocates a given number $n$ of PRBs to a given user depending on: \textit{i}) the class of services he belongs to (i.e.,the transmission rate he requires) and \textit{ii}) his position in the cell (i.e., the perceived radio conditions). Without loss of generality, we assume that there is just one class of service with a required transmission rate denoted by $C^*$.\\

A user located at distance $x$ from $\site$ decodes the signal only if the metric \textquotedblleft Signal to Interference plus Noise Ratio (SINR)\textquotedblright $~\Theta(x)= \frac{P x^{-2b}/a}{I+\sigma^2}$ is above a threshold $\Theta^*=\Theta(R)$, where $I$ is the received co-channel interference and $\sigma^2$ is the thermal noise power. For performance analysis purpose, SINR $\Theta(x)$ is often mapped to the user throughput by a link level curve. To simplify calculation, we use hereafter the upper bound of the well known Shannon's formula for MIMO system $Tx \times Rx $, with $Tx$ and $Rx$ are respectively the number of transmit and receive antennas. Hence, the throughput of a user located at distance $x$ from $\site$ is   

\begin{equation}
C(x)=\vartheta W\log_2\left(1+\Theta(x)\right), 
\label{Shannonmimo}
\end{equation}
with $\vartheta= min(Tx,Rx)$.\\

Then, the number of PRBs required by a user located at distance $x$ from $\site$ is
\begin{equation}
n(x)=\lceil \frac{C^*}{C(x)} \rceil \leq N,
\label{nbrprb}
\end{equation}
where $N= min(N_{max} , \lceil C^*/(\vartheta Wlog_2(1+\Theta^*))\rceil)$, $N_{max}$ is the maximum number of PRBs that a BS can allocate to a user (fixed by the operator) and $\lceil . \rceil$ stands for the Ceiling function.\\

It is obvious from (\ref{nbrprb}) that users are fairly scheduled because a user with bad radio conditions (with low value of $C(x)$) gets higher number of PRBs to achieve its transmission rate $C^*$.\\

Let $d_n$ be the distance from $\site$ that verifies, for all $x \in (d_{n-1},d_n]$, $n(x)=n$, with

\begin{equation}
n= \frac{C^*}{C(d_n)}  
\label{nbrprb2}
\end{equation}
is an integer and\\ 
\[
d_n=\left\{
\begin{array}{ll}
0 \ \mbox{if $n=0$,}  \\ 

%
%
\left[\frac{a(I+\sigma^2)}{P}(2^{\frac{C^*}{n \vartheta W}}-1)\right]^{\frac{-1}{2 b}}\ \mbox{otherwise,}\\
\\
\end{array}
\right.
\]

From (\ref{nbrprb2}), the cell $\site$ area can be divided into rings with radius $d_n$ such that for $1\leqslant n\leqslant N,~  0\leqslant d_{n-1}<d_n\leqslant R$. The area between the ring of radius $d_n$ and the ring of radius $d_{n-1}$ characterizes the region of the cell where users require $n$ PRBs to achieve the transmission rate $C^*$. Given that $d_n$ depends on the propagation parameter, it is worth to mention that there is a difference between $d_n$ values for outdoor and indoor environments. Thus to avoid confusion, we denote in the remainder, for indoor environment, the ring radius by $\tilde{d}_n$ and the propagation parameter by $\tilde{a}$. Finally, we define the cell throughput by the sum over all transmission rates of users:    
\begin{equation}
\tau= u C^*,
\label{cellth}
\end{equation}
with $u$ is recalled the average number of users inside $\site$ and expressed by (\ref{nbrusers}).\\

On the other hand, inter-cell interference is one of the main factors that compromise cellular network performance. The analysis of this factor level go through the SINR evaluation that depends on the geometry of the network as well as the distribution of users' locations. The analytical random models that can be found in literature, such as Homogeneous PPP, assume that BSs are randomly distributed according to a spatial point process. Thus, it becomes hard to estimate the interference level in each user location and only its distribution is determined; see for instance \cite{andrews2011tractable}.\\       

Besides, interference level estimation is of utmost importance in link adaptation procedure. In practical systems, the SINR is mapped to an indicator called Channel Quality Indicator (CQI) (e.g., 15 CQI indexes for LTE). This indicator is used by the BSs to determine the modulation and coding schemes (MCS) and consequently the transmission rate. Actually, the level of interference varies from one location to another in the same cell. Practically cell edge users experience high interference level compared to users that are close to the BS in the cell middle or cell center. To this purpose, one can consider three range of CQI indexes with a constant interference level for each range. The first range (i.e., low CQI indexes) stands for bad channel quality with high interference level, the second range (i.e., medium CQI indexes) stands for low interference level and the last range (i.e., high CQI indexes) refers to good channel quality with a negligible interference level.\\

 Furthermore, when interference level is non negligible, we use the notion of interference margin (IM) or Noise Rise in link budget. IM is defined as the increase in the thermal noise level caused by other-cell interference and it can be expressed in the linear scale as

\begin{equation}
IM=\frac{I+\sigma^2}{\sigma^2}
\label{NR}
\end{equation}                   

In the remainder of this study, we evaluate interference level by using three margins for each region of the studied cell. We consider three regions in $\site$ coverage area: the cell center that stands for the disk having a radius of $\frac{R}{3}$, the cell middle that represents the region between the disk B(0,$\frac{R}{3}$) and the disk B(0,$\frac{2R}{3}$). Finally, the cell edge refers to the region of the cell where the distance to $\site$ is above $\frac{2R}{3}$.\\

\section{Presentation of the dimensioning approach}

Dimensioning process consists in evaluating the required radio resources that allow to carry a forecast data traffic given a target QoS. The QoS can be measured by the congestion probability metric or even by a target average user throughput. The present approach assesses the congestion probability as a function of many key parameters, in particular, the number of PRBs $M$ and the cell throughput $\tau$. To characterize this congestion probability, we need to evaluate the total requested PRBs by all users. In the remainder of this section, we will state some analytical results regarding the explicit expression of the congestion probability under the system model presented in the previous sections.\\

\subsection{Qualification of the total number of requested PRBs}

As we have mentioned in section II, outdoor users are distributed along each road $L_j$ according to a linear PPP of intensity $\delta$. Now, if we consider a disk B(0,$d_n$) of radius $d_n$, the number of users in the portion of $L_j$ that lies inside B(0,$d_n$) is a Poisson random variable (this comes from the definition of the linear PPP) with mean $2 \delta \sqrt{d_n^2-r_j^2}$ (Pythagoras' theorem). Hence, conditionally on $\phi$, the mean number of users $\alpha_n(Y)$ inside B(0,$d_n$) is the sum over all roads $L_j$ that intersect with B(0,$d_n$) and it can be expressed by    

\begin{equation}
\alpha_{n}(Y)=2 \delta \sum_{j=1}^{Y}\mathds{1}_{(d_n>r_j)}\sqrt{d_n^2-r_j^2}.
\label{alphan}
\end{equation}

Moreover, the number of users in the portion of $L_j$ that lies between two rings B(0,$d_n$) and B(0,$d_{n-1}$) is also a Poisson random variable with parameter (i.e, the mean number of users)  $2\delta (\sqrt{d_n^2 - r_j^2}-\sqrt{d_{n-1}^2 - r_j^2})$. Finally, the mean number of users $\mu_n(Y)$ in all the roads that lie between the rings B(0,$d_n$) and B(0,$d_{n-1}$) can be expressed by

\begin{equation}
\mu_n(Y)=\alpha_n(Y) - \alpha_{n-1}(Y).
\end{equation}

 Similarly, the mean number of indoor users, that are distributed according to a spatial PPP of intensity $\kappa$ can be expressed by 
 
 \begin{equation}
 \tilde{\mu}_n= \kappa \pi (\tilde{d}_n^2 - \tilde{d}_{n-1}^2).
 \end{equation}

To qualify the number of requested PRBs by outdoor and indoor users, we consider two independent Poisson random variables denoted respectively by $X_n$ and $\tilde{X}_n$ with parameters $\mu_n(Y)$ and $\tilde{\mu}_n$. $X_n$ and $\tilde{X}_n$ represent the number of users (outdoor and indoor) that request $n$ PRBs with $1\leq n\leq N$.\\

%
%

Finally, we define the total number of requested PRBs in the cell as the sum of demanded PRBs by outdoor and indoor users in each ring. It can be expressed as

\begin{equation}
\Gamma= \mathcal{F} + \tilde{\mathcal{F}},    
\label{prbtot}
\end{equation}
where $\mathcal{F}=\sum_{n=1}^{N}nX_n$ and $\tilde{\mathcal{F}}=\sum_{n=1}^{N}n\tilde{X_n}$ are the total demanded PRBs by outdoor and indoor users respectively.\\

 The random variable $\Gamma$ is the sum of weighted Poisson random variables and it is called compound Poisson sum. The evaluation of its distribution requires extensive numerical simulation. It is important to mention that the parameter $\mu_n$ of $X_n$ depends on $Y$, which is a Poisson random variable. Hence all calculations should be done conditionally on $\phi$. The following proposition gives the explicit expression of the first-order moment (i.e., the mathematical expectation) of $\Gamma$.
 
 \begin{prop}
 	
 Let $\Gamma$ be a compound Poisson sum as in (\ref{prbtot}). Let $\phi$ be a PLP defined as in section II.A with $Y$ is the Poisson random variable that represents the number of roads that lie inside $\site$ coverage area.
 
 The first-order moment of $\Gamma$ is given by
 
 \begin{equation}
 \mathbb{E}(\Gamma)=\frac{4 \delta \omega}{3 R} \sum_{n=1}^{N}n \frac{d_n^3-d_{n-1}^3}{R} + \kappa \pi \sum_{n=1}^{N}n (\tilde{d}_n^2 - \tilde{d}_{n-1}^2),
 \end{equation}  	
 with  $\omega = 2 \pi \lambda R$ is the mathematical expectation of $Y$.	
 \label{meangamma}	
 \end{prop}
 
 \begin{proof}
 	See appendix \ref{proof0}.	
 \end{proof}
 

\subsection{Congestion probability and dimensioning approach}

The congestion probability, denoted by $\Pi$, is defined as the probability that the number of the total requested PRBs in the cell is greater than the available PRBs fixed by the operator. In other words, it measures the probability of failing to achieve an output number of PRBs $M$ required to guarantee a predefined quality of services:

\begin{equation}
\Pi(M,\tau)= \mathbb{P}(\Gamma \geq M).
\label{poutdef}
\end{equation}

 The following proposition gives the explicit expression of the congestion probability for a given process of users.
\begin{prop} 
	Let $\Lambda$ be a random variable such that $\Lambda=\sum_{n=1}^{N}nV_n$, with $V_n$ are Poisson random variables of intensity $w_n$. The probability that $\Lambda$ exceeds a threshold $M$ is 
	{
		\begin{align}
		\mathbb{P}(\Lambda\geq M)&= 1- \frac{1}{\pi} e^{-\sum_{n=1}^{N}w_n} \times\nonumber\\
		& \int_{0}^{\pi} e^{p_N(\theta)} \frac{\sin(\frac{M\theta}{2})}{\sin(\frac{\theta}{2})}\cos(\frac{M-1}{2}-q_N(\theta)) d\theta,
		\label{theee}
		\end{align}
	}
	where
	\[
	p_N(\theta)=\sum_{n=1}^{N}w_n \cos(n\theta)\text{ and } \\
	q_N(\theta)=\sum_{n=1}^{N}w_n \sin(n\theta).\\
	\]
	\label{theorem1}
\end{prop}

\begin{proof}
	See appendix \ref{proof1}.	
\end{proof}

This formula is valid for every process of user distribution including the spatial PPP which represents here the distribution of indoor users. The congestion probability $\mathbb{P}(\tilde{\mathcal{F}}\geq M)$ in this case can be explicitly determined by taking $w_n=\tilde{\mu}_n$ and using $\sum_{n=1}^{N}\tilde{\mu_n}= \kappa \pi R^2$ in (\ref{theee}). Similarly, for outdoor users modeled by Cox point process conditionally on the PLP $\phi$, proposition \ref{theorem1} remains valid with $w_n=\mu_n(Y)$ and $\sum_{n=1}^{N}\mu_n(Y)=\alpha_N(Y)$. The explicit expression of the congestion probability $\mathbb{P}(\mathcal{F}\geq M)$ in this case is calculated by averaging over the PLP $\phi$.\\

Moreover, from the superposition theorem of Poisson process, the congestion probability considering the combination of outdoor and indoor users is calculated by applying proposition \ref{theorem1} to the random variable $\Gamma= \sum_{n=1}^{N}nV_n$, with $V_n=X_n+\tilde{X}_n$ is a Poisson random variable having a parameter $w_n=\mu_n(Y) + \tilde{\mu}_n$.\\

The congestion probability expressions above can be developed even further by introducing a mathematical tool from combinatorial analysis called the exponential Bell polynomials \cite{roman1984exponential} and \cite{mihoubi2008bell}. This tool is widely used for the evaluation of integrals and alternating sums. In appendix \ref{bellreminder}, we introduce some key results of Bell Polynomials. \\

The following proposition gives the expression of the congestion probability as a function of the exponential complete Bell Polynomials.

\begin{prop} 
	Let $\Lambda$ be a random variable such that $\Lambda=\sum_{n=1}^{N}nV_n$, with $V_n$ are Poisson random variables of intensity $w_n$. Let $x_j$ be defined as
	
	\[
	x_j=\left\{
	\begin{array}{ll}
	w_j j! \ \ \ \ \ \mbox{if $1\leq j \leq N$,}  \\ 
	
	\\ 
	
	0  \ \ \ \ \  \ \ \ \ \ \ \ \ \ \ \ \ \ \ \ \mbox{otherwise.} \\
	
	\\
	\end{array}
	\right.
	\]

	 The probability that $\Lambda$ exceeds a threshold $M$ can be expressed as a function of the exponential complete Bell polynomials by 
	
		\begin{equation}
		\mathbb{P}(\Lambda \geq M)=1- H \sum_{k=0}^{M-1}\frac{B_k(x_1,...,x_k)}{k!} 
		\label{theo2}
		\end{equation}
		with $H=e^{-\sum_{n=1}^{N}w_n}$.
	\label{theorem2}
\end{prop}

\begin{proof}
	See appendix \ref{proof2}.	
\end{proof}
 
Now, to derive the expression of the congestion probability, we apply proposition \ref{theorem2} to the random variable $\Gamma$ defined in (\ref{prbtot}) as the superposition of two independent discrete random variables $\mathcal{F}$ and $\tilde{\mathcal{F}}$. $\Gamma$ can be written as
\begin{equation}
\Gamma=\sum_{n=1}^{N} n V_n,
\end{equation} 

with $V_n=X_n+\tilde{X}_n$ is a Poisson random variable of parameter $w_n=\mu_n(Y)+\tilde{\mu}_n$. Hence, by using proposition \ref{theorem2}, the congestion probability conditionally on $\phi$ (PLP) can be expressed as

\begin{equation}
P(\Gamma \geq M | \phi)= 1-H \sum_{k=0}^{M-1}\frac{B_k(x_1,...,x_k)}{k!},
\end{equation}
 
 where $x_j=(\mu_j(Y)+\tilde{\mu}_j)/j!$ and $H=e^{(-\alpha_N(Y) - \kappa \pi R^2)}$.\\

Once again, the final expression of the congestion probability is calculated by averaging over the PLP $\phi$ as 

\begin{equation}
\Pi(M,\tau) = \mathbb{E}_\phi[\mathbb{P}(\Gamma \geq M|\phi)].
\label{averaging}
\end{equation}

 Once we have the expression of the congestion probability, we set a target value $\Pi^*$ and then, the required number of PRBs $M$ is written as a function of $\tau$ through the implicit equation $\Pi(M,\tau)=\Pi^*$. The output $M$ of the implicit function constitutes the result of the dimensioning process.

\section{Numerical results}

For numerical purpose, we consider a cell of radius $R=0.7km$ with a transmit power level $P=60dBm$ (corresponds to $43dBm$ from the transmitter power amplifier and $17dBm$ for the antenna gain of the transmitter) and operating in a bandwidth of $20MHz$. The downlink thermal noise power including the receiver noise figure is calculated for $20MHz$ and set to $\sigma^2=-93dBm$. The propagation parameter is set to $a=130dB$ for outdoor users and to $\tilde{a}=166dB$ for indoor users. The path loss exponent is considered to be $2b=3.5$. We assume also that we have 8Tx antennas in the BS and 2Rx antennas in users' terminals. So, the number of possible transmission layers is at most 2.\\ 

\begin{figure}[tb]
	\centering
	\includegraphics[height=6.5cm,width=9.4cm]{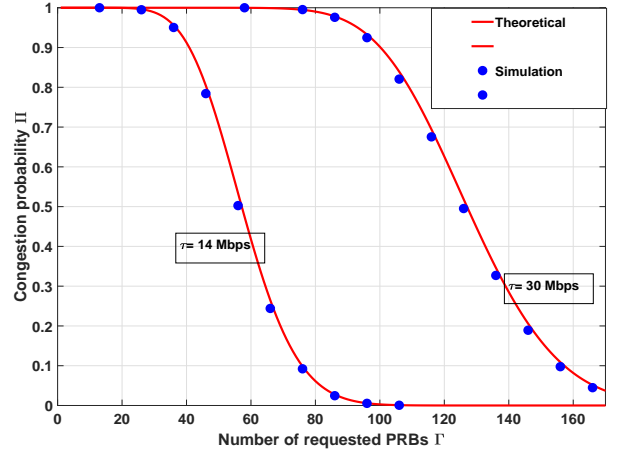}
	\caption{Congestion probability theoretical vs simulation for two values of $\tau$}
	\label{ths}
\end{figure}

In Fig. \ref{ths}, we simulate the described model in MATLAB for two values of cell throughput $\tau=14Mbps$ and $\tau=30Mbps$. We notice that the explicit expression of the congestion probability fits the empirical one obtained by using Monte-Carlo simulations. Moreover, it is obvious that an increase in cell throughput $\tau$ generates an increase of the congestion probability because $\tau$ is related to the number of users in the cell and depends on 3 intensities: outdoor users' intensity $\delta$, roads' intensity $\lambda$ and indoor users' intensity $\kappa$. When those intensities increase, the number of the required PRBs by users in the cell coverage area increases, thus the system experiences high congestion. An other important factor that can impact system performance is the path loss exponent. The variations of this parameter has tremendous effect on the congestion probability: when $2b$ goes up, radio conditions become worse and consequently the number of demanded PRBs to guarantee the required QoS increases.\\

 \begin{figure}[tb]
 	\centering
 	\includegraphics[height=6.5cm,width=9.4cm]{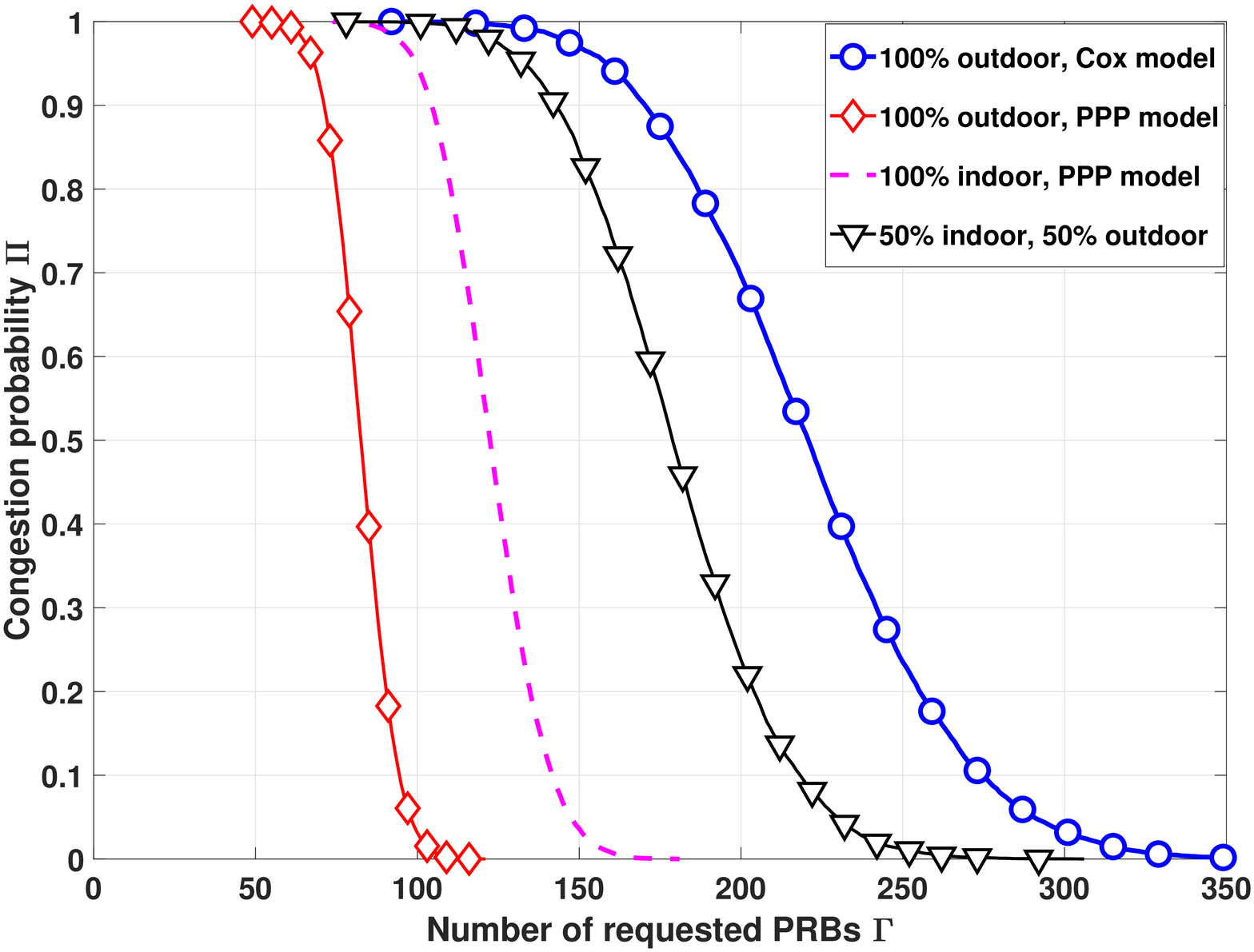}
 	\caption{Comparison between different user distributions}
 	\label{outdoorvsindoor}
 \end{figure}

To see how the random distribution of users impacts performance, we firstly depict in Fig. \ref{outdoorvsindoor} the congestion probability, considering only outdoor users in a random system of roads according to Cox process with roads intensity $\lambda=9km/km^2$ and users' intensity $\delta=6 users/km$, and secondly we compare it to the congestion probability of a spatial PPP outdoor users model with an equivalent intensity of $\lambda \delta=54 users/km^2$. We observe that the number of requested PRBs by users is always higher, for every target value of the congestion probability, when users are modeled by Cox process driven by PLP. In other words, even if the mean number of users in the cell is the same, the random tessellation of roads i.e., the geometry of the area covered by the cell has a significant impact on performance. Also, one can notice that if we consider a Cox model with high roads intensity, users appear to be distributed every where in the cell as in spatial PPP model with higher intensity. In this case, Cox process driven by PLP can be approximated by a spatial PPP.\\

Also in Fig. \ref{outdoorvsindoor}, we compare the congestion probability of indoor users modeled according to a spatial PPP and the one of outdoor users modeled according to a spatial PPP having the same intensity. We notice that indoor users required more PRBs than outdoor users and this comes from the difference between outdoor and indoor environment. Actually, signal propagation in indoor environment suffers from high attenuation and delay factors because of the presence of obstacles such as buildings and walls. Hence, indoor users always experience high path loss and bad performance in terms of SINR, which means that they need always more PRBs than outdoor users to achieve a required transmission rate.

%


\begin{figure}[tb]
	\centering
	\includegraphics[height=6.5cm,width=9.4cm]{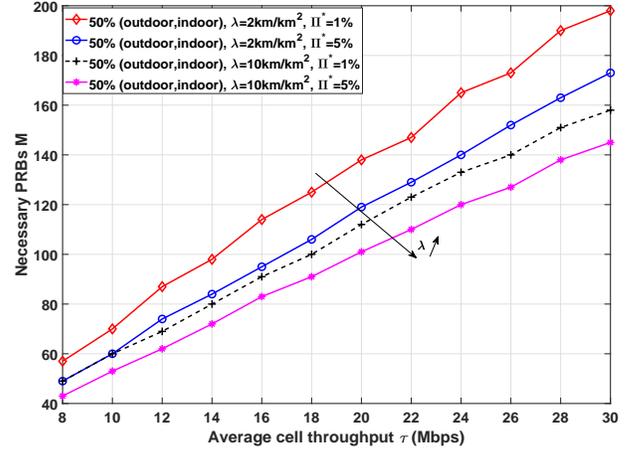}
	\caption{Required PRB $M$ as a function of cell Throughput $\tau$, for fixed transmission rate $C^*=500 kbps$.}
	\label{ths2}
\end{figure}

During resource dimensioning process, the operator starts by defining a target congestion probability that can be tolerated for a given service. For different traffic forecasts, the number of PRBs is set to ensure that the congestion probability never exceeds its target value. Fig. \ref{ths2} shows the number of required PRBs that the operator should make it available, when the expected cell throughput is known, for two target values of the congestion probability ($\Pi^*=1\%$ and $\Pi^*=5\%$) and for two road intensities ($\lambda=2km/km^2$ and $\lambda=10km/km^2$) with a fixed transmission rate of $500kbps$. We can observe that for each forecast cell throughput value, the threshold number of resources required in the cell decreases when road intensity increases. For instance, when $\lambda$ increases from $2km/km^2$ to $10km/km^2$ (i.e., from 9 expected roads to 44), the number of the dimensioned PRBs decreases by $32$, for the same cell throughput value $\tau=25 Mbps$. Also, for a given value of $\tau$, we can notice from (\ref{cellth}) that the user intensity on roads $\delta$ is inversely proportional to roads' intensity $\lambda$. Thus for fixed $\tau$, if $\lambda$ increases, $\delta$ decreases and consequently the number of required PRBs decreases.\\

\begin{figure}[tb]
	\centering
	\includegraphics[height=6.5cm,width=9.4cm]{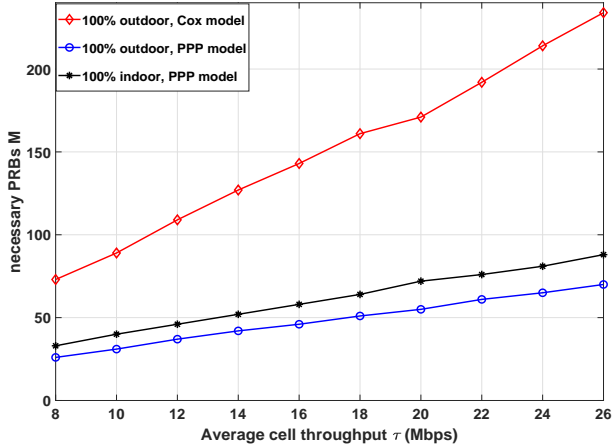}
	\caption{Dimensioned PRBs comparison: outdoor users with Cox model, outdoor users with spatial PPP and indoor users with spatial PPP, for a fixed transmission rate $C^*=500kbps$}
	\label{ths3}
\end{figure}

 Moreover, in Fig. \ref{ths3} we compare the dimensioning results for 3 models: Outdoor users according to Cox process driven by PLP, outdoor users according to a spatial PPP model and indoor users with spatial PPP model (having the same intensities). We notice that the number of dimensioned PRBs for outdoor users is always higher when users are modeled according to Cox process driven by PLP than spatial PPP model. Also, we can see that indoor users need more PRBs than outdoor users (when the both are modeled by the spatial PPP) which is in agreement with the previous results. Besides, we have mentioned previously that when $\lambda$ is very high, the distribution of users becomes similar to the one of a spatial PPP. Thus, with a spatial PPP model, one can have small values of the dimensioned PRBs, which is optimistic compared to the real geometry of the area covered by a cell in dense urban environment, where more PRBs are required to guarantee the desired quality of services.\\

To see interference impact on the dimensioning process, we divide, as we have mentioned previously, the cell into 3 regions: cell center with a radius of R/3, cell middle represented by the ring between R/3 and 2R/3 and cell edge characterized by a distance from the BS that exceeds 2R/3. Each region of the cell experiences a given level of interference evaluated in terms of IM (Interference Margin or Noise Rise). Cell edge users always experience high interference level and IM is set to be $15dB$. In cell middle we consider an interference margin of $8dB$, whereas in the cell center where users perceive good radio conditions, the interference margin is set to $IM=1dB$.\\

\begin{figure}[tb]
	\centering
	\includegraphics[height=6.5cm,width=9.4cm]{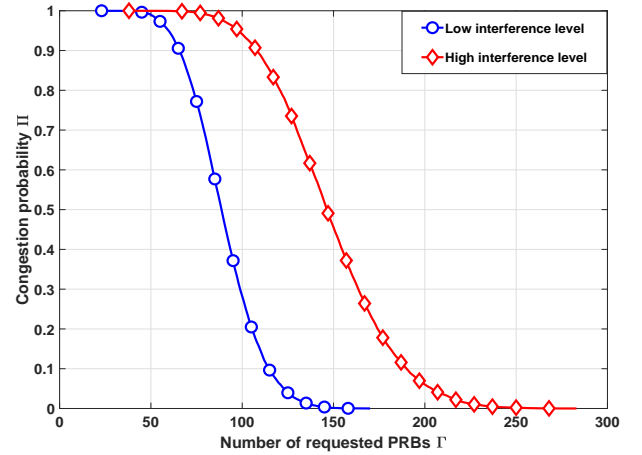}
	\caption{Interference impact ($\tau=30Mbps$).}
	\label{interference1}
\end{figure}

\begin{figure}[tb]
	\centering
	\includegraphics[height=6.5cm,width=9.4cm]{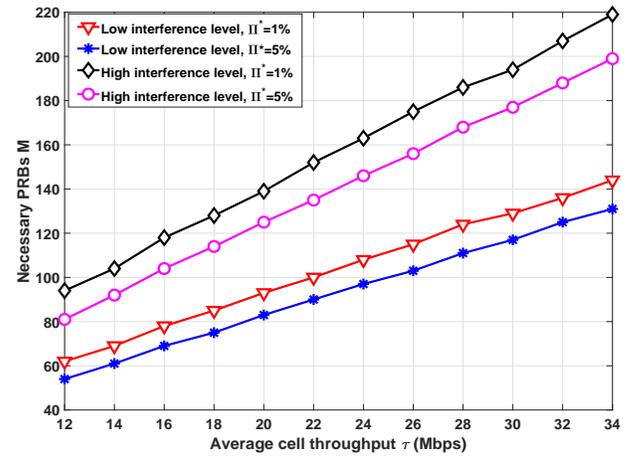}
	\caption{Interference impact on dimensioned PRBs M.}
	\label{ths4}
\end{figure}

Fig. \ref{interference1} shows the congestion probability in a noise-limited scenario (Interference level is neglected) and its comparison with the one where interference is taken in consideration as we have described above. We consider a scenario with $50\%$ of outdoor users modeled according to Cox process driven by the PLP and $50\%$ of indoor users modeled according to a spatial PPP, with an average cell throughput of 30Mbps and a fixed transmission rate of 500kbps. As expected, interference has a tremendous impact on the number of required PRBs. For instance, when the target congestion probability is set to $5\%$, the number of required PRBs increases by almost 80 because of the presence of interference. Similarly in Fig. \ref{ths4}, we plot the dimensioning curves i.e., the threshold number of PRBs in the cell as a function of the forecast average cell throughput, for a noise-limited environment and an environment with interference. As we can see, the number of PRBs that the operator should make it available is higher when interference impact is considered. For instance, for a forecast average cell throughput of 26Mbps and a target QoS $\Pi^*=5\%$, the number of dimensioned PRBs increases by almost 50 PRBs when the three interference margins are considered.\\ 

Besides, interference level varies from one location to another in the same cell. Practically cell edge users experience high interference level compared to users that are close to the BS in the cell middle or cell center. Fig. \ref{ths5} shows a comparison between resource dimensioning results for the three regions of the cell: cell center, cell middle and cell edge. As we can observe, the high demand on PRBs comes especially from cell edge users that perceive bad radio conditions because of the far distance from the BS and the presence of interference. Hence, for a predicted average cell throughput, the number of dimensioned PRBs should be set always by considering a probable presence of traffic hotspots at the cell edge.\\

\begin{figure}[tb]
	\centering
	\includegraphics[height=6.5cm,width=9.4cm]{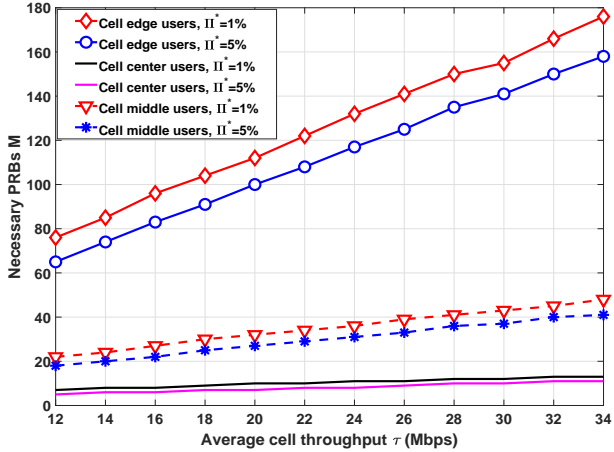}
	\caption{Comparison between dimensioned PRBs M for cell edge, cell middle and cell center users.}
	\label{ths5}
\end{figure}

Dimensioning phase is very important because it gives the operators a vision on how they should manage the available spectrum. If the dimensioned number of resources exceeds the available one, the operator can for instance:
\begin{itemize}
	\item aggregate fragmented spectrum resources into a single wider band in order to increase the available PRBs,
	\item activate capacity improvement features like carrier aggregation or dual connectivity between 5G and legacy 4G networks in order to delay investment on the acquisition of new spectrum bands,
	\item change the TDD (Time Division Duplexing), configuration to relieve the congested link,
	\item or even buy new spectrum bands.
\end{itemize}

\section{Conclusions}

In this paper, we have presented a resource dimensioning model for OFDM based systems that can be applied also for scalable OFDM based 5G NR interface. We have considered two spatial random distributions in order to distinguish between outdoor users distributed along a random system of roads in a typical cell coverage area (Cox Point Process driven by PLP) and indoor users distributed in buildings according to the widely used spatial PPP. The comparison between the two spatial distributions showed that results are more optimistic when spatial PPP is used. Also, we have shown that the geometry of the area covered by a cell can impact the results. Moreover, we have derived an analytical model to qualify the number of required PRBs in a typical cell with two explicit formulas of the congestion probability. Also, we have established an implicit relationship between the required resources and the forecast traffic given a target congestion probability. This relationship translates the dimensioning problem that an operator can perform to look for the amount of necessary spectrum resources to satisfy a predefined QoS. Finally, a comparison between an interfered environment and a noise-limited one has been provided. Besides, we have shown that the high requirement in terms of radio resources comes from cell edge users that perceive bad radio conditions.


\appendices
\section{Proof of proposition \ref{meangamma}}
\label{proof0}

Let $\Gamma$ be defined as in (\ref{prbtot}). $X_n$ and $\tilde{X}_{n}$ are two Poisson random variables with parameters $\mu_n(Y)$ and $\tilde{\mu}_n$. The random variable $X_n$ is dependent on the PLP $\phi$ i.e., depends on $Y$.  Hence, the mathematical expectation of $\Gamma$ can be written as
{
	
    \begin{align}
    \mathbb{E}(\Gamma)&= \mathbb{E}_{\phi}(\Gamma|\phi) \nonumber\\
       &=\sum_{n=1}^{N}n \mathbb{E}_{\phi}(\mu_n(Y)) + \sum_{n=1}^{N} n \tilde{\mu}_n 
     \label{st1}
    \end{align}

}

To evaluate equation (\ref{st1}), we need to calculate first the mathematical expectation of $\mu_n(Y)$. Let $\omega=2 \pi \lambda R$ be the mathematical expectation of the Poisson random variable $Y$. $\mathbb{E}_{\phi}(\mu_n(Y))$ can be expressed as  

{
	\begin{align}
      \mathbb{E}_{\phi}(\mu_n(Y))&= 2 \delta \sum_{k=1}^{+\infty} \frac{\omega^k e^{-\omega}}{k!} \sum_{j=1}^{k}\mathbb{E}_{r_j}\left[\mathds{1}_{(d_n>r_j)}\sqrt{d_n^2-r_j^2}\right]- \nonumber\\
      &\mathbb{E}_{r_j}\left[\mathds{1}_{(d_{n-1}>r_j)}\sqrt{d_{n-1}^2-r_j^2}\right]
    \label{st2} 
	\end{align}
	
}

$\{r_j\}$ follow a uniform distribution in the disk of radius $R$ representing the whole cell coverage area. Thus
\begin{equation}
\mathbb{E}_{r_j}\left[\mathds{1}_{(d_n>r_j)}\sqrt{d_n^2-r_j^2}\right]=\frac{2}{R^2}\int_{0}^{d_n}\sqrt{d_n^2-r^2}rdr.
\end{equation}

Finally, by using a change of variable $x=r^2$ and the expression of $\tilde{\mu}_n$, we get the result of proposition \ref{meangamma}, which completes the proof.
  
\section{Proof of proposition \ref{theorem1}}
\label{proof1}

To prove proposition \ref{theorem1}, we calculate at first the moment generating function (i.e., Z-Transform) $f(z)$ of the discrete random variable $\Lambda$.

{
	\begin{align}
	f(z)=\mathbb{E}(z^{\Lambda}) &= \sum_{k=0}^{+\infty} z^k \mathbb{P}(\Lambda =k)\nonumber\\
	&=  \prod_{n=1}^{N}\sum_{k=0}^{+\infty} z^{nk}\mathbb{P}(V_n=k). 
	\label{p1} 
	\end{align}
}
Since $V_n$ is a Poisson random variable with parameter $w_n$, (\ref{p1}) is simplified to 
\begin{equation}		
f(z)=e^{-\sum_{n=1}^{N}w_n}e^{\sum_{n=1}^{N}z^n w_n},
\label{p3} 
\end{equation}

It is obvious that $f$ is analytic on $\mathbb{C}$ and in particular inside the unit circle $\mathcal{C}$. Cauchy's integral formula gives then the coefficients of the expansion of $f$ in the neighborhood of $z=0$:
\begin{equation}		
\mathbb{P}(\Lambda =k)= \frac{1}{2 \pi i}\int_{\mathcal{C}}\frac{f(z)}{z^{k+1}}dz.
\label{cauchy} 
\end{equation}
In (\ref{cauchy}), replacing $f$ by its expression (\ref{p3}) and parameterizing $z$ by $e^{i\theta}$ lead to
\begin{equation}		
\mathbb{P}(\Lambda =k)= \frac{1}{2 \pi} e^{-\sum_{n=1}^{N}w_n} \int_{0}^{2 \pi} \frac{e^{\sum_{n=1}^{N}w_n e^{i n \theta}}}{e^{i k \theta}}d\theta.
\label{p5} 
\end{equation}

Since the congestion probability is defined by the CCDF (Complementary Cumulative Distribution Function) of $\Lambda$, then
{
	\begin{align}
	&\mathbb{P}(\Lambda \geq M) =1-\sum_{k=0}^{M-1} \mathbb{P}(\Lambda=k)\nonumber \\
	&= 1-\frac{1}{2 \pi} e^{-\sum_{n=1}^{N}w_n} \int_{0}^{2 \pi} e^{\sum_{n=1}^{N}w_n e^{i n \theta}} \sum_{k=0}^{M-1} e^{-i k \theta} d\theta.
	\label{p6} 
	\end{align}
}
The sum inside the right hand integral of (\ref{p6}) can be easy calculated to get the explicit expression of (\ref{theee}) after some simplifications.

\section{Key background on the exponential Bell Polynomials}
\label{bellreminder}
Exponential Bell polynomials $B_{p}$ are obtained from their generating function 
\begin{equation}
e^{\sum_{j=1}^{+\infty}x_j \frac{t^j}{j!}}= \sum_{p=0}^{+\infty} \frac{t^p}{p!} B_{p}(x_1,x_2,...x_{p}).
\label{bellpartial}
\end{equation}
 and have the following combinatorial expression 

\begin{equation}
B_{p}(x_1,x_2,....,x_p)=\sum_{k_1+2k_2+...=p}^{} \frac{p!}{k_1! k_2!....} (\frac{x_1}{1!})^{k_1} (\frac{x_2}{2!})^{k_2}.... \nonumber
\label{bell}
\end{equation} 

Also, if we consider the matrix $A_p =(a_{i,j})_{1 \leq i,j \leq p}$ defined by

\[
\left\{
\begin{array}{ll}
a_{i,j}= {p-i \choose j-i}x_{j-i+1} \ \ \ \ \ \mbox{if $i \leq j$,}  \\ 

\\ 

a_{i,i-1}=-1  \ \ \ \ \ \ \ \ \ \ \ \ \ \mbox{if $i\geq 2$ } \\

\\

a_{i,j}=0  \ \ \ \ \ \ \ \ \ \ \ \ \  \ \ \ \ \ \mbox{if $i\geq j+2$,} \\

\\
\end{array}
\right.
\]
then 
\begin{equation}
B_p(x_1,..,x_p)= Det(A_p).
\label{bellbinomials}
\end{equation}

%
%
%

For instance, the first few Bell Polynomials are given by
{
	\begin{align}
	&B_0= 1 \nonumber \\
	&B_1(x_1)=x_1 \nonumber \\
	&B_2(x_1,x_2)=x_1^2+x_2 \nonumber \\
	&B_3(x_1,x_2,x_3)=x_1^3+3 x_1 x_2 + x_3 \nonumber \\
	&B_4(x_1,x_2,x_3,x_4)=x_1^4 + 6x_1^2x_2 + 4 x_1 x_3 + 3 x_2^2 + x_4  \nonumber \\
	&\vdots \nonumber
	\label{fewbells} 
	\end{align}
}

Bell polynomials are multi-variable Sheffer sequence and then satisfy the binomial type relation:

\begin{equation}
B_p(x_1+y_1,..,x_p+y_p)=\sum_{i=0}^{p}{{p} \choose{i}} B_{p-i}(x_1,..,x_{p-i}) B_i(y_1,..,y_i).
\label{bellbinomial}
\end{equation}

\section{Proof of proposition \ref{theorem2}}
\label{proof2}

By using the definition of $x_j$, the Z-Transform of $\Lambda$ given in appendix \ref{proof1} by equation (\ref{p3})  becomes

\begin{equation}
f(z)= H e^{\sum_{j=0}^{+ \infty}z^j\frac{x_j}{j!}},
\label{be1}
\end{equation}
with $H=e^{-\sum_{n=1}^{N}w_n} $.

The second exponential term in (\ref{be1}) can be evaluated by using the generating function of the complete Bell Polynomials given in equation (\ref{bellpartial}), it follows that

\begin{equation}
f(z)= H \sum_{p=0}^{+ \infty} \frac{z^p}{p!} B_p(x_1,...,x_p)
\label{be2}
\end{equation}

On the other hand, by using the definition of Z-Transform of $\Lambda$ and the Taylor expansion of $f(z)$ in 0, it follows that  

\begin{equation}
\mathbb{P}(\Lambda= p)= \frac{H}{p!} B_p(x_1,...,x_p).
\label{be3}
\end{equation}

Finally, from the definition of the CCDF (Complementary Cumulative Distribution Function), we get

\begin{equation}
\mathbb{P}(\Lambda \geq M)=1-\sum_{k=0}^{M-1} \mathbb{P}(\Lambda= k),
\label{be4}
\end{equation}

which completes the proof.

\balance

\bibliographystyle{IEEEtran}
\bibliography{IEEEabrv,TelecomReferences}
\end{document}